\documentclass[a4paper]{jpconf}

\usepackage{epic,eepic}
\usepackage[mathscr]{eucal} 
\usepackage{psfrag}     

\usepackage{amsmath}
\usepackage{amsthm}

\usepackage[dvips]{graphicx,color}  

\usepackage{amssymb}
\usepackage{epic,eepic}
\usepackage{amscd}
\usepackage{longtable}
\usepackage{array}

\newtheorem{theorem}{Theorem}
\newtheorem{lemma}[theorem]{Lemma}

\newtheorem{proposition}[theorem]{Proposition}

\theoremstyle{definition}

\newtheorem{remark}[theorem]{Remark}

\newcommand{\geh}{\mathfrak{g}}
\newcommand{\ot}{\otimes}

\newcommand{\Z}{{\mathbb Z}}

\newcommand{\C}{{\mathbb C}}
\newcommand{\Q}{{\mathbb Q}}

\begin{document}
\title{Tetrahedron equation and quantum $R$ matrices for $q$-oscillator representations}

\author{A Kuniba$^1$ and M Okado$^2$}

\address{$^1$ Graduate School of Arts and Sciences, University of Tokyo, Komaba, 153-8902, Japan}

\address{$^2$ Department of Mathematics, Osaka City University, 
Osaka, 558-8585, Japan}

\begin{abstract}
We review and supplement the recent result by the authors 
on the reduction of the three dimensional $R$ (3d $R$) 
satisfying the tetrahedron equation to the quantum $R$ matrices for 
the $q$-oscillator representations of 
$U_q(D^{(2)}_{n+1})$, $U_q(A^{(2)}_{2n})$ and 
$U_q(C^{(1)}_{n})$.
A new formula for the 3d $R$ and 
a quantum $R$ matrix for $n=1$ are presented and 
a proof of the irreducibility of the tensor product of 
the $q$-oscillator representations is detailed.
\end{abstract}

\section{Introduction}
This paper is a summary and supplement of the recent result 
\cite{KO3} by the authors, which is motivated by the earlier works \cite{S,BS,KS}. 
The tetrahedron equation (\ref{TE}) \cite{Zam80} is a three dimensional 
generalization of the Yang-Baxter equation \cite{Bax}.
In \cite{KS} a new prescription was proposed to reduce it to the 
Yang-Baxter equation 
$R_{1,2}R_{1,3}R_{2,3} = R_{2,3}R_{1,3}R_{1,2}$ 
by using the special boundary vectors defined by (\ref{cv}) and (\ref{vb}).
Applied to 
a particular solution of the tetrahedron equation (3d $L$ operator \cite{BS}),
the reduction was shown \cite{KS} to give 
the quantum $R$ matrices for the spin representations \cite{O}. 

In \cite{KO3} a similar reduction was studied for the 
distinguished solution of the tetrahedron equation which we call 3d $R$.
The 3d $R$ was obtained as the intertwiner of the quantum coordinate ring 
$A_q(sl_3)$ \cite{KV}, (The original formula on p194 therein contains a misprint.) 
and was found later also in a different setting \cite{BS}. 
They were shown to coincide  
and to constitute the solution of the 3d reflection equation in \cite{KO1}.
See \cite[App. A]{KO3} for more detail.
The main result of \cite{KO3} was the identification of the 
reduction of the 3d $R$ with the quantum $R$ matrices for 
the quantum affine algebras 
$U_q= U_q(D^{(2)}_{n+1})$, $U_q(A^{(2)}_{2n})$ and 
$U_q(C^{(1)}_{n})$.
Their relevant representations turned out to be  
new infinite dimensional ones which we called 
the {\em $q$-oscillator representations}.
There are two kinds of boundary vectors, which curiously 
correspond to the choices of the above three algebras.
See Remark \ref{re:dyn}. 

This paper contains a summary of these results and 
a few supplements.
The formula (\ref{Rex}) for the 3d $R$ and (\ref{new}) for the 
quantum $R$ matrix for $n=s=t=1$ case are new.
Section \ref{sec:4} recollects a proof of the 
irreducibility of the tensor product of the $q$-oscillator representations
whose detail was omitted in \cite{KO3}.
The result for $n=1$ was reported earlier in \cite{KO2}.
More recently it has been shown that 
the $q$-oscillator representations \cite{KO3} quoted in 
Prop. \ref{pr:repD}-\ref{pr:repC} here actually 
factor through a homomorphism from $U_q$ to the 
$n$ fold tensor product of the
{\it $q$-oscillator algebra} \cite{KOS}.

Throughout the paper we assume that $q$ is generic and 
use the following notations:
\begin{align*}
&(z;q)_m = \prod_{k=1}^m(1-z q^{k-1}),\;\;
(q)_m = (q; q)_m,\;\;
\binom{m}{k}_{\!\!q}= \frac{(q)_m}{(q)_k(q)_{m-k}},
\;\;[m]_q= \frac{q^m-q^{-m}}{q-q^{-1}},
\end{align*}
where the $q$-binomial is to be understood as zero
unless $0 \le k \le m$.
$[m]_{q^t}$ with $t=1$ will simply be denoted by $[m]$.
\section{Reducing the tetrahedron equation to the Yang-Baxter equation}
\subsection{General scheme using boundary vectors}
Let  $F$ be a vector space and 
$R \in \mathrm{End}(F^{\otimes 3})$.
Consider the tetrahedron equation:
\begin{align}\label{TE}
R_{1,2,4}R_{1,3,5}R_{2,3,6}R_{4,5,6}=
R_{4,5,6}R_{2,3,6}R_{1,3,5}R_{1,2,4} \in \mathrm{End}(F^{\otimes 6})
\end{align}
where $R_{i,j,k}$ acts as $R$ on the 
$i,j,k$ th components from the left in $F^{\otimes 6}$.

We recall the prescription which produces  
an infinite family of solutions to the Yang-Baxter equation 
from a solution to the tetrahedron equation based on 
special boundary vectors \cite{KS}.
First we regard (\ref{TE}) as a one-site relation,
and extend it to the $n$-site version.
Let $\overset{\alpha_i}{F},
\overset{\beta_i}{F},
\overset{\gamma_i}{F}$ be the copies of $F$,
where $\alpha_i, \beta_i$ and $\gamma_i\,(i=1,\ldots, n)$
are just labels.
Renaming the spaces $1,2,3$ by them gives
$R_{\alpha_i, \beta_i, 4}
R_{\alpha_i, \gamma_i, 5}
R_{\beta_i, \gamma_i, 6}R_{4,5,6}=
R_{4,5,6}
R_{\beta_i, \gamma_i, 6}
R_{\alpha_i, \gamma_i, 5}
R_{\alpha_i, \beta_i, 4}$
for each $i$.
Thus for any $i$ one can carry $R_{4,5,6}$ through  
$R_{\alpha_i, \beta_i, 4}
R_{\alpha_i, \gamma_i, 5}
R_{\beta_i, \gamma_i, 6}$ to the left 
reversing it into
$R_{\beta_i, \gamma_i, 6}
R_{\alpha_i, \gamma_i, 5}
R_{\alpha_i, \beta_i, 4}$.
Applying this $n$ times leads to
\begin{equation}\label{TEn}
\begin{split}
&\bigl(R_{\alpha_1, \beta_1, 4}
R_{\alpha_1, \gamma_1, 5}
R_{\beta_1, \gamma_1, 6}\bigr)
\cdots 
\bigl(R_{\alpha_n, \beta_n, 4}
R_{\alpha_n, \gamma_n, 5}
R_{\beta_n, \gamma_n, 6}\bigr)R_{4,5,6}\\
&=
R_{4,5,6}
\bigl(R_{\beta_1, \gamma_1, 6}
R_{\alpha_1, \gamma_1, 5}
R_{\alpha_1, \beta_1, 4}
\bigr)
\cdots 
\bigl(R_{\beta_n, \gamma_n, 6}
R_{\alpha_n, \gamma_n, 5}
R_{\alpha_n, \beta_n, 4}
\bigr).
\end{split}
\end{equation}
This is an equality in 
$\mathrm{End}(\overset{\boldsymbol \alpha}{F}\otimes 
\overset{\boldsymbol\beta}{F}\otimes 
\overset{\boldsymbol\gamma}{F}\otimes 
\overset{4}{F}\otimes 
\overset{5}{F}\otimes 
\overset{6}{F})$,
where 
${\boldsymbol\alpha}=(\alpha_1,\ldots, \alpha_n)$
is the array of labels and 
$\overset{\boldsymbol \alpha}{F}= 
\overset{\alpha_1}{F}\otimes \cdots \otimes 
\overset{\alpha_n}{F}\,(= F^{\otimes n})$.
The notations  
$\overset{\boldsymbol\beta}{F}$ and $\overset{\boldsymbol\gamma}{F}$
should be understood similarly.

Next we introduce special boundary vectors.
Suppose one has a vector
$|\chi_s(x)\rangle \in F$
depending on a variable $x$ such that 
its tensor product
\begin{equation}\label{cv}
|\chi_s(x,y)\rangle
=|\chi_s(x)\rangle \otimes |\chi_s(xy)\rangle
\otimes |\chi_s(y)\rangle
\in F\otimes F\otimes F
\end{equation}
satisfies the relation
\begin{equation}\label{RX}
R |\chi_s(x,y)\rangle   = |\chi_s(x,y)\rangle.
\end{equation}
The index $s$ is put to distinguish possibly more than one such vectors.
Suppose there exist vectors in the dual space
\begin{equation*}
\langle \chi_s(x,y)|
=\langle \chi_s(x)| \otimes
\langle \chi_s(xy)| \otimes
\langle \chi_s(y)|
\in F^*\otimes F^*\otimes F^*
\end{equation*}
having the similar property
\begin{equation}\label{XR}
\langle \chi_s(x,y) |R =\langle \chi_s(x,y) | .
\end{equation}
Then evaluating (\ref{TEn}) between 
$\langle \chi_s(x,y) |$ and $|\chi_t(1,1)\rangle$,
one obtains
\begin{align}\label{sdef}
S_{\boldsymbol{\alpha, \beta}}(z)=\varrho^{s,t}(z)\langle \chi_s(z)|
R_{\alpha_1, \beta_1, 3}
R_{\alpha_2, \beta_2, 3}\cdots
R_{\alpha_n, \beta_n, 3}
|\chi_t(1)\rangle
\in \mathrm{End}
(\overset{\boldsymbol\alpha}{F}\otimes \overset{\boldsymbol\beta}{F}),
\end{align}
where $\varrho^{s,t}(z)$ is inserted to control the normalization.
The composition of $R$ and 
matrix elements are taken for the space signified by $3$.
One may simply write it as
$S(z) \in \mathrm{End}(F^{\otimes n} \otimes F^{\otimes n})$
dropping the dummy labels.
The  $S(z)$ depends on $s$ and $t$ 
although they have been temporarily suppressed.
It follows from (\ref{TEn}), (\ref{RX}) and (\ref{XR}) that 
$S(z)$ satisfies the Yang-Baxter equation:
\begin{align}\label{sybe}
S_{\boldsymbol{\alpha, \beta}}(x)
S_{\boldsymbol{\alpha, \gamma}}(xy)
S_{\boldsymbol{\beta,\gamma}}(y)
=
S_{\boldsymbol{\beta,\gamma}}(y)
S_{\boldsymbol{\alpha, \gamma}}(xy)
S_{\boldsymbol{\alpha, \beta}}(x)
\in \mathrm{End}(
\overset{\boldsymbol\alpha}{F}\otimes
\overset{\boldsymbol\beta}{F}\otimes
\overset{\boldsymbol\gamma}{F}).
\end{align}

\subsection{A realization of the scheme}\label{ss:ce}
We focus on the solution $R$ of the tetrahedron equation  
mentioned in the introduction.
Take $F$ to be an infinite dimensional space 
$F = \bigoplus_{m\ge 0}\Q(q)|m\rangle$
with the dual  
$F^\ast = \bigoplus_{m\ge 0}\Q(q)\langle m|$
having the bilinear pairing 
$\langle l |m\rangle = (q^2)_m\delta_{l,m}$.
Then the 3d $R$ is given by
\begin{align}
&R(|i\rangle \otimes |j\rangle \otimes |k\rangle) = 
\sum_{a,b,c\ge 0} R^{a,b,c}_{i,j,k}
|a\rangle \otimes |b\rangle \otimes |c\rangle,\label{Rabc}\\
&R^{a,b,c}_{i,j,k} = \delta^{a+b}_{i+j}\delta^{b+c}_{j+k}
\sum_{\lambda+\mu=b}(-1)^\lambda
q^{ik+b+\lambda(c-a)+\mu(\mu-i-k-1)}
\binom{\lambda+a}{a}_{\!q^2}\binom{i}{\mu}_{\!q^2},\label{Rex}
\end{align}
where $\delta^i_j=1$ if $i=j$ and $0$ otherwise.
The sum is over $\lambda, \mu \in \Z_{\ge 0}$ such that 
$\lambda+\mu=b$ with the further condition $\mu\le i$.
It satisfies $(q^2)_a(q^2)_b(q^2)_c\,R^{a,b,c}_{i,j,k} 
= (q^2)_i(q^2)_j(q^2)_k\,R^{i,j,k}_{a,b,c}$ \cite[eq.(A.1)]{KO3}.
The formula (\ref{Rex})
is simpler than \cite[eq.(2.10)]{KO3}.
Its derivation will be given elsewhere. 

The two boundary vectors satisfying (\ref{RX}) and (\ref{XR})
are known \cite{KS} and given by
\begin{align}
\langle \chi_s(z)| = \sum_{m\ge 0}\frac{z^m}{(q^{s^2})_m}\langle sm|,
\quad
|\chi_s(z)\rangle = \sum_{m\ge 0}\frac{z^m}{(q^{s^2})_m}|sm\rangle
\quad (s=1,2).
\label{vb}
\end{align}

Given two boundary vectors, 
one can construct four families of solutions to the Yang-Baxter equation
$S(z) = S^{s,t}(z)= S^{s,t}(z,q)$ $(s,t=1,2)$ by 
(\ref{sdef}) by substituting  (\ref{Rex}) and (\ref{vb}).
Each family consists of the solutions 
labeled with $n \in \Z_{\ge 1}$. 
They are the matrices acting on $F^{\otimes n} \otimes F^{\otimes n}$ 
whose elements read
\begin{align}
&S^{s,t}(z)\bigl(|{\bf i}\rangle \otimes |{\bf j}\rangle\bigr)
= \sum_{{\bf a},{\bf b}}
S^{s,t}(z)^{{\bf a},{\bf b}}_{{\bf i},{\bf j}}
|{\bf a}\rangle \otimes |{\bf b}\rangle,\label{sact}\\
&S^{s,t}(z)^{{\bf a},{\bf b}}_{{\bf i},{\bf j}}
=\varrho^{s,t}(z)\!\!\!\sum_{c_0, \ldots, c_n\ge 0} 
\frac{z^{c_0}(q^2)_{sc_0}}{(q^{s^2})_{c_0}(q^{t^2})_{c_n}}
R^{a_1, b_1, sc_0}_{i_1, j_1, c_1}
R^{a_2, b_2, c_1}_{i_2, j_2, c_2}\cdots
R^{a_{n\!-\!1}, b_{n\!-\!1}, c_{n\!-\!2}}_{i_{n\!-\!1}, j_{n\!-\!1}, c_{n\!-\!1}}
R^{a_n, b_n, c_{n\!-\!1}}_{i_n, j_n, tc_n},\label{sabij}
\end{align}
where $|{\bf a}\rangle = 
|a_1\rangle \otimes \cdots \otimes |a_n\rangle \in F^{\otimes n}$
for ${\bf a} = (a_1,\ldots, a_n) \in (\Z_{\ge 0})^n$, etc.
By Applying \cite[eq.(A.1)]{KO3} to (\ref{sabij}) it is straightforward to show
\begin{align}\label{s21}
S^{t,s}(z)^{{\bf a},{\bf b}}_{{\bf i},{\bf j}}/\varrho^{t,s}(z)
=\left(\prod_{r=1}^n\frac{z^{\frac{1}{t}j_r}(q^2)_{i_r}(q^2)_{j_r}}
{z^{\frac{1}{t}b_r}(q^2)_{a_r}(q^2)_{b_r}}\right)
S^{s,t}(z^{\frac{s}{t}})^{\overline{\bf i},
\overline{\bf j}}_{\overline{\bf a},\overline{\bf b}}
/\varrho^{s,t}(z^{\frac{s}{t}}),
\end{align}
where $\overline{\bf a} = (a_n,\ldots, a_1)$ is the reverse array of  
${\bf a} = (a_1,\ldots, a_n)$ and similarly for 
$\overline{\bf b}, \overline{\bf i}$ and $\overline{\bf j}$.
Henceforth we shall only consider
$S^{1,1}(z), S^{1,2}(z)$ and $S^{2,2}(z)$ in the rest of the paper.
The matrix elements $R^{a,b,c}_{i,j,k}$ (\ref{Rex}) and 
$S^{s,t}(z)^{{\bf a},{\bf b}}_{{\bf i},{\bf j}}$ 
(\ref{sabij}) are depicted as follows:
\[
\begin{picture}(200,75)(-210,-40)

\put(-220,0){
\put(-80,0){$R^{a,b,c}_{i,j,k}=$}
\put(0,0){\vector(0,1){15}}\put(0,0){\line(0,-1){15}}
\put(-2,19){$\scriptstyle{b}$} \put(-2,-22){$\scriptstyle{j}$}

\put(0,0){\line(3,1){20}}\put(0,0){\vector(-3,-1){20}}
\put(22,6){$\scriptstyle{k}$} \put(-28,-9){$\scriptstyle{c}$}

\put(0,0){\line(-3,1){20}}\put(0,0){\vector(3,-1){20}}
\put(23,-9){$\scriptstyle{a}$}\put(-27,6){$\scriptstyle{i}$}
}

\put(-160,0){$S^{s,t}(z)^{{\bf a},{\bf b}}_{{\bf i},{\bf j}}=$}

\put(3,1){\vector(-3,-1){73}}

\put(-115,-27){$\scriptstyle{\langle \chi_s(z) |}$}

\put(-83,-27){$\scriptstyle{sc_0}$}
 
\put(-48,-16){\vector(0,1){16}}\put(-48,-16){\line(0,-1){16}}
\put(-48,-16){\vector(3,-1){16}} \put(-48,-16){\line(-3,1){16}}
\put(-51,3){$\scriptstyle{b_1}$}
\put(-72,-10){$\scriptstyle{i_1}$}
\put(-31,-26){$\scriptstyle{a_1}$}
\put(-50,-39){$\scriptstyle{j_1}$}

\put(-30,-15){$\scriptstyle{c_1}$}

\put(-15,-5){\vector(0,1){13}}\put(-15,-5){\line(0,-1){13}}
\put(-15,-5){\vector(3,-1){13}}\put(-15,-5){\line(-3,1){13}}
\put(-36,0){$\scriptstyle{i_2}$}
\put(-18,11){$\scriptstyle{b_2}$}
\put(-17,-25){$\scriptstyle{j_2}$}
\put(0,-13){$\scriptstyle{a_2}$}

\put(2,-4){$\scriptstyle{c_2}$}

\multiput(5.1,1.7)(3,1){7}{.} 
\put(6,2){
\put(21,7){\line(3,1){30}}
\put(36,12){\vector(0,1){12}}\put(36,12){\line(0,-1){12}}
\put(36,12){\vector(3,-1){12}}\put(36,12){\line(-3,1){12}}
\put(15,16){$\scriptstyle{i_n}$}
\put(50,5){$\scriptstyle{a_n}$}
\put(33,27){$\scriptstyle{b_n}$}
\put(34,-7){$\scriptstyle{j_n}$}

\put(17,1){$\scriptstyle{c_{n-1}}$}

\put(53,17){$\scriptstyle{tc_n}$}
}
 
\put(75,19){$\scriptstyle{|\chi_t(1) \rangle}$}
 
 \end{picture}
 \]

Due to $\delta$ factors in (\ref{Rex}), 
$S^{s,t}(z)$ obeys 
the conservation law
\begin{align}\label{claw}
S^{s,t}(z)^{{\bf a},{\bf b}}_{{\bf i},{\bf j}}=0\;\;
\text{unless}\;\;
{\bf a}+{\bf b} = {\bf i} + {\bf j}
\end{align}
and the sum (\ref{sabij}) is constrained by the 
$n$ conditions
$b_1+sc_0 = j_1+c_1, \ldots, b_n+c_{n-1}=j_n+tc_n$
leaving effectively a {\em single} sum.
For $(s,t)=(2,2)$, they further enforce a parity constraint
\begin{align}\label{d22}
S^{2,2}(z)^{{\bf a},{\bf b}}_{{\bf i},{\bf j}}=0
\;\;
\text{unless}\;\;
|{\bf a}| \equiv |{\bf i}|,\;\; |{\bf b}| \equiv |{\bf j}| \;\mod 2,
\end{align}
where $|{\bf a}|=a_1+\cdots + a_n$, etc.
Thus we have a direct sum decomposition
\begin{align}
S^{2,2}(z)& 
= S^{+,+}(z)\oplus S^{+,-}(z)\oplus S^{-,+}(z)\oplus S^{-,-}(z),
\label{22pm}\\
S^{\epsilon_1,\epsilon_2}(z)
& \in \mathrm{End}
\bigl((F^{\otimes n})_{\epsilon_1}\otimes 
(F^{\otimes n})_{\epsilon_2}\bigr),
\qquad
(F^{\otimes n})_{\pm} =
\bigoplus_{{\bf a} \in (\Z_{\ge 0})^n,\,
(-1)^{|{\bf a}|}=\pm 1}\Q(q)|{\bf a}\rangle.\label{fpm}
\end{align} 
We dare allow the coexistence of somewhat confusing notations 
$S^{s,t}(z)$ and $S^{\epsilon_1,\epsilon_2}(z)$
expecting that they can be properly distinguished from the context.
(A similar warning applies to $\varrho^{s,t} (z)$ in the sequel.)
We choose the normalization factors as
\begin{align}\label{rst}
\varrho^{1,1} (z)= 
\frac{(z; q)_\infty}{(-zq; q)_\infty},\;\;
\varrho^{1,2} (z)= 
\frac{(z^2; q^2)_\infty}{(-z^2q; q^2)_\infty},\;\;
\varrho^{\epsilon_1,\epsilon_2}(z)= 
\Bigl(\frac{(z; q^4)_\infty}{(zq^2; q^4)_\infty}
\Bigr)^{\epsilon_1\epsilon_2}.
\end{align}
Then the matrix elements of $S^{1,1}(z), S^{1,2}(z)$ 
and $S^{\epsilon_1,\epsilon_2}(z)$ are 
rational functions of $q$ and $z$.

\subsection{Example}\label{ss:ex}
Let us present an explicit form of the matrix element 
(\ref{sabij}) for $n=1$.
It was worked out earlier in \cite[Prop.2]{KO2} by using a formula 
for $R^{a,b,c}_{i,j,k}$ different from (\ref{Rex}). 
For simplicity we concentrate on the case $s=t=1$ and 
write $S^{s,t}(z)^{{\bf a},{\bf b}}_{{\bf i},{\bf j}}$ 
as $S(z)^{a,b}_{i,j}$ with $a,b,i,j \in \Z_{\ge 0}$.
A direct calculation using (\ref{Rex}) and (\ref{rst}) leads to 
\begin{equation}\label{new}
\begin{split}
&S(z)(|i\rangle\otimes |j \rangle) 
= \sum_{a,b\ge 0}S(z)^{a,b}_{i,j}|a\rangle \otimes |b\rangle,\qquad
S(z)^{a,b}_{i,j}
= z^{a-i}\frac{(q^2)_i(q^2)_j}{(q^2)_a(q^2)_b}S(z)^{i,j}_{a,b},\\
&S(z)^{a,b}_{i,j} =\delta^{a+b}_{i+j} 
\sum_{\lambda, \mu}
(-1)^\lambda q^{j(1-a)+\mu(\mu-1)}\binom{j}{\lambda}_{\!\!q^2}
\binom{\lambda+i}{b}_{\!\!q^2}
\frac{(-q;q)_{i-a}(z;q)_{a+\lambda-\mu}}
{(-zq,q)_{i+\lambda-\mu}} \quad
(0\! \le\! a \!\le\! i).
\end{split}
\end{equation}
The last sum is over $\lambda, \mu \in \Z_{\ge 0}$
such that $\lambda + \mu = j$ and $\lambda+i \ge b$.
Thus it is actually a single sum over $\max(0,b-i) \le \lambda \le j$.
The formula (\ref{new}) is simpler than \cite[eq.(2.19)]{KO2}.
From our main Theorem \ref{th:main} it follows that 
$S^{a,b}_{i,j}(z=1) = \delta^a_j\delta^b_i$, 
which is consistent with the above result.

\section{Quantum $R$ matrices for $q$-oscillator representations}\label{sec:R}

The Drinfeld-Jimbo quantum affine algebras without derivation 
$U_q=U_q(D^{(2)}_{n+1})$, $U_q(A^{(2)}_{2n})$ and 
$U_q(C^{(1)}_{n})$ are the Hopf algebras 
generated by $e_i, f_i, k^{\pm 1}_i\, (0 \le i \le n)$ satisfying the relations
\cite{D,Ji}:
\begin{equation*}
\begin{split}
&k_i k^{-1}_i = k^{-1}_i k_i = 1,\;[k_i, k_j]=0,\;
k_ie_jk^{-1}_i = q_i^{a_{ij}}e_j,\;
k_if_jk^{-1}_i = q_i^{-a_{ij}}f_j,\;
[e_i, f_j]=\delta_{ij}\frac{k_i-k^{-1}_i}{q_i-q^{-1}_i},\\
&\sum_{\nu=0}^{1-a_{ij}}(-1)^\nu
e^{(1-a_{ij}-\nu)}_i e_j e_i^{(\nu)}=0,
\quad
\sum_{\nu=0}^{1-a_{ij}}(-1)^\nu
f^{(1-a_{ij}-\nu)}_i f_j f_i^{(\nu)}=0\;\;(i\neq j),
\end{split}
\end{equation*}
where $e^{(\nu)}_i = e^\nu_i/[\nu]_{q_i}!, \,
f^{(\nu)}_i = f^\nu_i/[\nu]_{q_i}!$ with 
$[\nu]_q = [\nu]_q [\nu-1]_q \cdots [1]_q$.
The Cartan matrix $(a_{ij})_{0 \le i,j \le n}$ \cite{Kac} is given by
$a_{i,j} = 2\delta_{i,j}-\max((\log q_j)/(\log q_i),1)\delta_{|i-j|,1}$.
The data $q_i$ is specified 
above the corresponding vertex $i\, (0 \le i \le n)$ 
in the Dynkin diagrams:
\begin{align*}
\begin{picture}(120,60)(0,-15)
\put(-25,0){
\put(50,25){$D^{(2)}_{n+1}$}
\multiput( 0,0)(20,0){3}{\circle{6}}
\multiput(100,0)(20,0){2}{\circle{6}}
\multiput(23,0)(20,0){2}{\line(1,0){14}}
\put(83,0){\line(1,0){14}}
\multiput( 2.85,-1)(0,2){2}{\line(1,0){14.3}} 
\multiput(102.85,-1)(0,2){2}{\line(1,0){14.3}} 
\multiput(59,0)(4,0){6}{\line(1,0){2}} 
\put(10,0){\makebox(0,0){$<$}}
\put(110,0){\makebox(0,0){$>$}}
\put(0,-5){\makebox(0,0)[t]{$0$}}
\put(20,-5){\makebox(0,0)[t]{$1$}}
\put(40,-5){\makebox(0,0)[t]{$2$}}
\put(100,-5){\makebox(0,0)[t]{$n\!\! -\!\! 1$}}
\put(120,-6.5){\makebox(0,0)[t]{$n$}}
\put(3,18){\makebox(0,0)[t]{$q^{\frac{1}{2}}$}}
\put(20,13){\makebox(0,0)[t]{$q$}}
\put(40,13){\makebox(0,0)[t]{$q$}}
\put(100,13){\makebox(0,0)[t]{$q$}}
\put(123,18){\makebox(0,0)[t]{$q^{\frac{1}{2}}$}}
}
\end{picture}
\begin{picture}(120,60)(0,-15)
\put(0,0){
\put(50,25){$A^{(2)}_{2n}$}
\multiput( 0,0)(20,0){3}{\circle{6}}
\multiput(100,0)(20,0){2}{\circle{6}}
\multiput(23,0)(20,0){2}{\line(1,0){14}}
\put(83,0){\line(1,0){14}}
\multiput( 2.85,-1)(0,2){2}{\line(1,0){14.3}} 
\multiput(102.85,-1)(0,2){2}{\line(1,0){14.3}} 
\multiput(59,0)(4,0){6}{\line(1,0){2}} 
\put(10,0){\makebox(0,0){$<$}}
\put(110,0){\makebox(0,0){$<$}}
\put(0,-5){\makebox(0,0)[t]{$0$}}
\put(20,-5){\makebox(0,0)[t]{$1$}}
\put(40,-5){\makebox(0,0)[t]{$2$}}
\put(100,-5){\makebox(0,0)[t]{$n\!\! -\!\! 1$}}
\put(120,-6.5){\makebox(0,0)[t]{$n$}}
\put(3,18){\makebox(0,0)[t]{$q^{\frac{1}{2}}$}}
\put(20,12.5){\makebox(0,0)[t]{$q$}}
\put(40,12.5){\makebox(0,0)[t]{$q$}}
\put(100,12.5){\makebox(0,0)[t]{$q$}}
\put(122,16){\makebox(0,0)[t]{$q^2$}}
}
\end{picture}
\begin{picture}(120,60)(10,-15)
\put(35,0){
\put(50,25){$C^{(1)}_{n}$}
\multiput( 0,0)(20,0){3}{\circle{6}}
\multiput(100,0)(20,0){2}{\circle{6}}
\multiput(23,0)(20,0){2}{\line(1,0){14}}
\put(83,0){\line(1,0){14}}
\multiput( 2.85,-1)(0,2){2}{\line(1,0){14.3}} 
\multiput(102.85,-1)(0,2){2}{\line(1,0){14.3}} 
\multiput(59,0)(4,0){6}{\line(1,0){2}} 
\put(10,0){\makebox(0,0){$>$}}
\put(110,0){\makebox(0,0){$<$}}
\put(0,-5){\makebox(0,0)[t]{$0$}}
\put(20,-5){\makebox(0,0)[t]{$1$}}
\put(40,-5){\makebox(0,0)[t]{$2$}}
\put(100,-5){\makebox(0,0)[t]{$n\!\! -\!\! 1$}}
\put(120,-6.5){\makebox(0,0)[t]{$n$}}
\put(2,15){\makebox(0,0)[t]{$q^2$}}
\put(20,12){\makebox(0,0)[t]{$q$}}
\put(40,12){\makebox(0,0)[t]{$q$}}
\put(100,12){\makebox(0,0)[t]{$q$}}
\put(122,15){\makebox(0,0)[t]{$q^2$}}
}
\end{picture}
\end{align*}
We employ the coproduct $\Delta$ of the form 
$\Delta(k^{\pm 1}_i) = k^{\pm 1}_i\otimes k^{\pm 1}_i$,
$\Delta(e_i) = 1\otimes e_i + e_i \otimes k_i$ and 
$\Delta(f_i)= f_i\otimes 1 + k^{-1}_i\otimes f_i$.

\subsection{$q$-oscillator representations}\label{ss:qor}
We introduce representations of $U_q$
on the tensor product of the Fock space ${\hat F}^{\otimes n}$
or $F^{\otimes n}$, where
${\hat F}= \bigoplus_{m\ge 0}\C(q^{\frac{1}{2}})|m\rangle$
is a slight extension of the coefficient field of $F$.
They all factor through an algebra homomorphism 
from $U_q$ to the $q$-oscillator algebra
as shown in \cite[Prop. 2.1]{KOS}.
As in the previous section 
we write the elements of ${\hat F}^{\otimes n}$ as
$|{\bf m}\rangle =  |m_1\rangle \otimes  \cdots \otimes |m_n \rangle \in 
{\hat F}^{\otimes n}$
for ${\bf m} = (m_1,\ldots, m_n) \in (\Z_{\ge 0})^n$
and describe the changes in ${\bf m}$ by the vectors 
${\bf e}_i = (0,\ldots,0,\overset{i}{1},0,\ldots, 0) \in  \Z^n$.
In the following propositions $\kappa = \frac{q+1}{q-1}$ and 
$x$ is a nonzero parameter.
\begin{proposition}\label{pr:repD}
The following defines an irreducible $U_q(D^{(2)}_{n+1})$ module structure on 
${\hat F}^{\otimes n}$.
\begin{align*}
e_0|{\bf m}\rangle &= x|{\bf m}+{\bf e}_1\rangle,\\
f_0|{\bf m}\rangle  &= i \kappa [m_1] x^{-1}|{\bf m}-{\bf e}_1\rangle,\\
k_0|{\bf m}\rangle  &= -i q^{m_1+\frac{1}{2}}|{\bf m}\rangle,\\
e_j|{\bf m}\rangle &= [m_j]|{\bf m}-{\bf e}_j+{\bf e}_{j+1}\rangle
\quad(1\le  j  \le n-1),\\
f_j|{\bf m}\rangle &= [m_{j+1}]|{\bf m}+{\bf e}_j-{\bf e}_{j+1}\rangle
\quad(1\le  j  \le n-1),\\
k_j|{\bf m}\rangle &= q^{-m_{j}+m_{j+1}}|{\bf m}\rangle
\quad(1\le  j  \le n-1),\\
e_n|{\bf m}\rangle & =  i\kappa[m_n]|{\bf m}-{\bf e}_n\rangle,\\
f_n|{\bf m}\rangle &=|{\bf m}+{\bf e}_n\rangle,\\
k_n|{\bf m}\rangle &= iq^{-m_n-\frac{1}{2}}|{\bf m}\rangle.
\end{align*}
\end{proposition}
\begin{proposition}\label{pr:repA}
The following defines an irreducible $U_q(A^{(2)}_{2n})$ module structure on 
${\hat F}^{\otimes n}$.
\begin{align*}
e_0|{\bf m}\rangle &= x|{\bf m}+{\bf e}_1\rangle,\\
f_0|{\bf m}\rangle  &= i\kappa [m_1]x^{-1}|{\bf m}-{\bf e}_1\rangle,\\
k_0|{\bf m}\rangle  &= -iq^{m_1+\frac{1}{2}}|{\bf m}\rangle,\\
e_j|{\bf m}\rangle &= [m_j]|{\bf m}-{\bf e}_j+{\bf e}_{j+1}\rangle
\quad(1\le  j  \le n-1),\\
f_j|{\bf m}\rangle &= [m_{j+1}]|{\bf m}+{\bf e}_j-{\bf e}_{j+1}\rangle
\quad(1\le  j  \le n-1),\\
k_j|{\bf m}\rangle &= q^{-m_{j}+m_{j+1}}|{\bf m}\rangle
\quad(1\le  j  \le n-1),\\
e_n|{\bf m}\rangle & = \frac{[m_n][m_n-1]}{[2]^2}
|{\bf m}-2{\bf e}_n\rangle,\\
f_n|{\bf m}\rangle &=|{\bf m}+2{\bf e}_n\rangle,\\
k_n|{\bf m}\rangle &=-q^{-2m_n-1}|{\bf m}\rangle.
\end{align*}
\end{proposition}

\begin{proposition}\label{pr:repC}
The following defines an irreducible $U_q(C^{(1)}_{n})$ module structure 
on $(F^{\otimes n})_+$ and $(F^{\otimes n})_-$ defined in (\ref{fpm}).
\begin{align*}
e_0|{\bf m}\rangle &= x|{\bf m}+2{\bf e}_1\rangle,\\
f_0|{\bf m}\rangle  &= \frac{[m_1][m_1-1]}{[2]^2}
 x^{-1}|{\bf m}-2{\bf e}_1\rangle,\\
k_0|{\bf m}\rangle  &= -q^{2m_1+1}|{\bf m}\rangle,\\
e_j|{\bf m}\rangle &= [m_j]|{\bf m}-{\bf e}_j+{\bf e}_{j+1}\rangle
\quad(1\le  j  \le n-1),\\
f_j|{\bf m}\rangle &= [m_{j+1}]|{\bf m}+{\bf e}_j-{\bf e}_{j+1}\rangle
\quad(1\le  j  \le n-1),\\
k_j|{\bf m}\rangle &= q^{-m_{j}+m_{j+1}}|{\bf m}\rangle
\quad(1\le  j  \le n-1),\\
e_n|{\bf m}\rangle & = \frac{[m_n][m_n-1]}{[2]^2}
|{\bf m}-2{\bf e}_n\rangle,\\
f_n|{\bf m}\rangle &=|{\bf m}+2{\bf e}_n\rangle,\\
k_n|{\bf m}\rangle &= -q^{-2m_n-1}|{\bf m}\rangle.
\end{align*}
\end{proposition}
We call these irreducible representations 
the {\em $q$-oscillator representations} of $U_q$.
For the twisted case 
$U_q(D^{(2)}_{n+1})$ and $U_q(A^{(2)}_{2n})$,
they are {\em singular} at $q=1$ because of the factor 
$\kappa$.

\subsection{Quantum $R$ matrices}\label{ss:qR}
Let $V= \hat{F}^{\otimes n}$ for 
$U_q(D^{(2)}_{n+1}),  U_q(A^{(2)}_{2n})$ and 
$V= F^{\otimes n}$ for $U_q(C^{(1)}_n)$.
First we consider $U_q(D^{(2)}_{n+1})$ and $U_q(A^{(2)}_{2n})$.
Let $V_x = {\hat F}^{\otimes n}[x,x^{-1}]$ 
be the representation space of $U_q$ 
in Propositions \ref{pr:repD} and \ref{pr:repA}.  
By the existence of the universal $R$ matrix \cite{D}
there exists an element  
$R \in \mathrm{End}(V_x\otimes V_y)$ 
such that 
\begin{align}\label{drrd}
\Delta'(g) R = R \Delta(g)\quad 
\forall g \in U_q
\end{align}
up to an overall scalar.
Here $\Delta'$ is the opposite coproduct 
defined by $\Delta' = P \circ \Delta$, where
$P(u \otimes v) = v \otimes u$ is the exchange of the components.
A little inspection of our representations shows that 
$R$ depends on $x$ and $y$ only through the ratio $z=x/y$.
Moreover $V_x\ot V_y$ is irreducible 
(\cite[Prop. 12]{KO3} and Sec. \ref{ss:irr} of this paper)
hence $R$ is determined only by postulating (\ref{drrd}) 
for $g=k_r,e_r$ and $f_r$ with $0 \le r \le n$.
Thus denoting the $R$ by $R(z)$, 
we may claim \cite{Ji} that it is determined by the conditions 
\begin{align}
(k_r \otimes k_r)R(z) &= R(z)(k_r\otimes k_r),\label{kR}\\
(e_r\otimes1 + k_r \otimes e_r) R(z) &= 
R(z)(1\otimes e_r + e_r\otimes k_r),\label{eR}\\
(1\otimes f_r + f_r \otimes k^{-1}_r) R(z) &= 
R(z)(f_r\otimes 1 + k^{-1}_r\otimes f_r)\label{fR}
\end{align}
for $0 \le r \le n$ up to an overall scalar.
We fix the normalization of $R(z)$ by 
\begin{align}\label{rnor}
R(z)(|{\bf 0}\rangle \otimes |{\bf 0}\rangle)
= |{\bf 0}\rangle \otimes |{\bf 0}\rangle,
\end{align}
where $|{\bf 0}\rangle \in {\hat F}^{\otimes n}$ is defined 
in the beginning of Section \ref{ss:qor} with 
${\bf 0}= (0,\ldots, 0)$.
We call the intertwiner $R(z)$ 
the {\em quantum $R$ matrix} for $q$-oscillator representation.
It satisfies the Yang-Baxter equation
\begin{align}\label{yber}
R_{12}(x)R_{13}(xy)R_{23}(y)
= R_{23}(y)R_{13}(xy)R_{12}(x).
\end{align}

Next we consider $U_q(C^{(1)}_n)$.
Denote by $V^{\pm}_x= (F^{\otimes n})_{\pm}[x,x^{-1}]$
the representation spaces in Proposition \ref{pr:repC}
and set 
$V_x = V^+_x \oplus V^-_x = F^{\otimes n}[x,x^{-1}]$.
See (\ref{fpm}) for the definition of $(F^{\otimes n})_{\pm}$.
We define the quantum $R$ matrix
$R(z)$ to be the direct sum 
\begin{align}\label{Rdeco}
R(z) = R^{+,+}(z)\oplus 
R^{+,-}(z)\oplus 
R^{-,+}(z)\oplus 
R^{-,-}(z),
\end{align}
where each 
$R^{\epsilon_1,\epsilon_2}(z)\in 
\mathrm{End}(V^{\epsilon_1}_x\otimes V^{\epsilon_2}_y)$ is
the quantum $R$ matrix with the 
normalization condition 
\begin{equation}\label{Rpmnor}
\begin{split}
&R^{+,+}(z) (|{\bf 0}\rangle \otimes |{\bf 0}\rangle)
=|{\bf 0}\rangle \otimes |{\bf 0}\rangle,\quad
\qquad\quad
R^{+,-}(z)(|{\bf 0}\rangle \otimes |{\bf e}_1\rangle)
=\frac{-iq^{1/2}}{1-z}|{\bf 0}\rangle \otimes |{\bf e}_1\rangle,\\
&R^{-,+}(z)(|{\bf e}_1\rangle \otimes |{\bf 0}\rangle)=
\frac{-iq^{1/2}}{1-z}|{\bf e}_1\rangle \otimes |{\bf 0}\rangle,\quad
R^{-,-}(z)(|{\bf e}_1\rangle \otimes |{\bf e}_1\rangle)
=\frac{z-q^2}{1-zq^2}
|{\bf e}_1\rangle \otimes |{\bf e}_1\rangle.
\end{split}
\end{equation}
The $R$ matrix $R(z)$
satisfies the Yang-Baxter equation (\ref{yber}).
In fact it is decomposed into the finer equalities 
($\epsilon_1, \epsilon_2, \epsilon_3=\pm$)
\begin{align*}
R^{\epsilon_1,\epsilon_2}_{12}(x)
R^{\epsilon_1,\epsilon_3}_{13}(xy)
R^{\epsilon_2,\epsilon_3}_{23}(y)
= 
R^{\epsilon_2,\epsilon_3}_{23}(y)
R^{\epsilon_1,\epsilon_3}_{13}(xy)
R^{\epsilon_1,\epsilon_2}_{12}(x).
\end{align*}

\subsection{Main theorem}\label{ss:main}
Define the operator $K$ acting on ${\hat F}^{\otimes n}$ by 
$K|{\bf m}\rangle = (-iq^{\frac{1}{2}})^{m_1+\cdots + m_n}|{\bf m}\rangle$.
Introduce the gauge transformed quantum $R$ matrix by
\begin{align}\label{gtr}
{\tilde R}(z) = (K^{-1}\otimes 1) R(z)(1\otimes K).
\end{align}
It is easy to see that ${\tilde R}(z)$ 
also satisfies the Yang-Baxter equation (\ref{yber}).

In Section \ref{ss:ce} the solutions $S^{s,t}(z)$  of 
the Yang-Baxter equation have been 
constructed from the 3d $R$  in 
(\ref{sact}), (\ref{sabij}) and (\ref{rst}).
In Section \ref{ss:qR}
the quantum $R$ matrices for $q$-oscillator representations
of $U_q(D^{(2)}_{n+1})$, 
$U_q(A^{(2)}_{2n})$ and 
$U_q(C^{(1)}_{n})$ 
have been defined. 
The next theorem, which is the main result of \cite{KO3}, 
states the precise relation between them.
(See (\ref{s21}) for $S^{2,1}(z)$.)
\begin{theorem}\label{th:main}
Denote by ${\tilde R}_{\mathfrak g}(z)$ the gauge transformed
quantum $R$ matrix (\ref{gtr}) for 
$U_q({\mathfrak g})$.
Then the following equalities hold:
\begin{align*}
S^{1,1}(z) = {\tilde R}_{D^{(2)}_{n+1}}(z),\quad
S^{1,2}(z) = {\tilde R}_{A^{(2)}_{2n}}(z),\quad
S^{2,2}(z) = {\tilde R}_{C^{(1)}_{n}}(z),
\end{align*}
where the last one means
$S^{\epsilon_1,\epsilon_2}(z) = \tilde{R}^{\epsilon_1,\epsilon_2}(z)$
between (\ref{22pm}) and (\ref{Rdeco}) with the 
gauge transformation (\ref{gtr}).
\end{theorem}

\begin{remark}\label{re:dyn}
Theorem \ref{th:main} suggests
the following correspondence between the boundary vectors 
(\ref{vb}) with the end shape of the Dynkin diagrams:

\begin{picture}(200,70)(-42,3)

\put(100,50){
\put(-0.4,0){\circle{6}}
\drawline(2,-1.8)(18,-2)
\drawline(2,1.8)(18,2)
\drawline(7,0)(13,-6)
\drawline(7,0)(13,6)
\put(-3,10){\small $0$}
\put(-50,-2){$\langle {\chi}_1(z) |$}
}

\put(100,18){
\put(-0.4,0){\circle{6}}
\drawline(2,-2)(18,-1.8)
\drawline(2,2)(18,1.8)
\drawline(7,6)(13,0)
\drawline(7,-6)(13,0)
\put(-3,10){\small $0$}
\put(-50,-2){$\langle {\chi}_2(z) |$}
}

\put(90,50){
\drawline(82,-2)(98,-1.8)
\drawline(82,2)(98,1.8)
\drawline(93,0)(87,-6)
\drawline(93,0)(87,6)
\put(100.4,0){\circle{6}}
\put(120,-2){$|\chi_1(1)\rangle$}
\put(98,10){\small $n$}
}

\put(90,18){
\drawline(82,-1.8)(98,-2)
\drawline(82,1.8)(98,2)
\drawline(93,6)(87,0)
\drawline(93,-6)(87,0)
\put(100.4,0){\circle{6}}
\put(120,-2){$|\chi_2(1)\rangle$}
\put(98,10){\small $n$}
}

\end{picture}

Consistently with Remark \ref{re:dyn},  
$S^{2,1}(z)$, which is reducible to $S^{1,2}(z^{1/2}) $ by (\ref{s21}),
is identified \cite{KOS} with the quantum $R$ matrix for $q$-oscillator 
representation of another $U_q(A^{(2)}_{2n})$ realized as  
the affinization of the classical part $U_q(B_n)$.
(Proposition \ref{pr:repA} corresponds to 
taking the classical part to be $U_q(C_n)$.)
As far as $\langle {\chi}_1(z) |$ and $|\chi_1(1)\rangle$ are 
concerned, the above correspondence agrees 
with the observation made in 
\cite[Remark 7.2]{KS} on the similar result concerning a 
3d $L$ operator.
With regard to $\langle {\chi}_2(z) |$ and $|\chi_2(1)\rangle$,
the relevant affine Lie algebras $A^{(2)}_{2n}$ and 
$C^{(1)}_n$ in this paper are the subalgebras 
of $B^{(1)}_{n+1}$ and $D^{(1)}_{n+2}$
in \cite[Theorem 7.1]{KS}
obtained by folding their Dynkin diagrams. 
\end{remark}

\section{Proof of the irreducibility of the tensor product}\label{sec:4}
In \cite{KO3} we gave a proof of the following proposition.
\begin{proposition}[Prop. 12 of \cite{KO3}] \label{pr:irred}
As a $U_q(D_{n+1}^{(2)})$ or $U_q(A_{2n}^{(2)})$ module $V_x\ot V_y$ 
is irreducible. As a $U_q(C_n^{(1)})$ module each $V_x^{\epsilon_1}\ot
V_y^{\epsilon_2}$ $(\epsilon_1,\epsilon_2=\pm)$ is irreducible.
\end{proposition}
Since the explanation there was not sufficient, we give the detailed proof here.
Let $\geh=D_{n+1}^{(2)},A_{2n}^{(2)}$ or $C_n^{(1)}$, $I=\{0,1,\ldots,n\}$, and for 
a subset $J$ of $I$ let $U_q(\geh_J)$ be the subalgebra of $U_q(\geh)$ generated
by $\{e_j,f_j,k_j^{\pm1}\mid j\in J\}$. Recall the vector $v_l$ \cite[Prop. 4]{KO3}
for $\geh=D_{n+1}^{(2)}$ and $v_l^\epsilon$ \cite[Prop. 5]{KO3} for $A_{2n}^{(2)},
C_n^{(1)}$.
\begin{proposition} \label{pr:w}
For $\geh=D_{n+1}^{(2)}$
\[
V^{\ot2}=\sum_{l=0}^\infty U_q(\geh_{I\setminus\{0\}})v_l,
\]
and for $\geh=A_{2n}^{(2)},C_n^{(1)}$
\[
V^{\epsilon_1}\ot V^{\epsilon_2}=\sum_{l=0\atop(-1)^l=\epsilon_1\epsilon_2}^\infty 
U_q(\geh_{I\setminus\{0\}})v_l^{\epsilon_1}.
\]
\end{proposition}
This is an immediate consequence of the following two lemmas. Set
$w_{l,k}=|k{\bf e}_{n-1}\rangle\ot |(l-k){\bf e}_n\rangle\in V^{\ot2}\;(l\ge0,0\le k\le l)$.

\begin{lemma}
For $\geh=D_{n+1}^{(2)}$ 
\[
w_{l,k}\in \sum_{j=0}^lU_q(\geh_{\{n-1,n\}})v_j,
\]
and for $\geh=A_{2n}^{(2)},C_n^{(1)}$
\[
w_{l,k}\in \sum_{j=0\atop j\equiv l\!\!\!\!
\pmod2}^lU_q(\geh_{\{n-1,n\}})v_j^\epsilon
\quad\mbox{where }\epsilon=(-1)^k.
\]
\end{lemma}

\begin{proof}
We treat the $\geh=D_{n+1}^{(2)}$ case first. Note that the set of vectors 
$B=\{v_l,f_nv_{l-1},\ldots,f_n^lv_0\}$ is linearly independent in the vector 
subspace spanned by $\{|k{\bf e}_n\rangle\ot|(l-k){\bf e}_n\rangle\mid 0\le k\le l\}$.
Hence, $B$ is also a basis and 
\begin{equation} \label{eq1}
w_{l,0}=|{\bf 0}\rangle\ot|l{\bf e}_n\rangle\in\sum_{j=0}^lU_q(\geh_{\{n\}})v_j.
\end{equation}
Next note that 
\[
w_{l,k}=(f_{n-1}f_n-q^{-1}f_nf_{n-1})w_{l-1,k-1}+i\frac{q^{l-k+1/2}}{[l-k+1]}f_{n-1}w_{l,k-1}
\quad(1\le k\le l).
\]
This relation together with \eqref{eq1} shows the result. 

Suppose now $\geh=A_{2n}^{(2)},C_n^{(1)}$. We compare 
$B'=\{v_l^\epsilon,f_nv_{l-1}^\epsilon,\ldots,f_n^{l'}v_{l-2l'}^\epsilon\}$
 ($l'=\lfloor l/2\rfloor\,(\epsilon=+),\,=\lfloor(l-1)2\rfloor\,(\epsilon=-)$) 
and the subspace spanned by
$\{|k{\bf e}_n\rangle\ot|(l-k){\bf e}_n\rangle\mid 0\le k\le l,(-1)^k=\epsilon\}$.
We have 
\begin{align*}
w_{l,0}=|{\bf 0}\rangle\ot|l{\bf e}_n\rangle&\in\sum_{j=0\atop j\equiv l\!\!\!\!
\pmod2}^lU_q(\geh_{\{n\}})v_j^+, \\
|{\bf e}_n\rangle\ot|(l-1){\bf e}_n\rangle&\in\sum_{j=0\atop j\equiv l\!\!\!\!
\pmod2}^lU_q(\geh_{\{n\}})v_j^-\quad(l\ge1).
\end{align*}
From
\begin{align*}
w_{1,1}&=f_{n-1}v_1^-,\quad w_{2,1}=\frac{q^{-1}}{[2]}v_0^-+\frac1{[2]^2}f_{n-1}v_2^-,\\
w_{l,1}&=-\frac{[l-1]}{[l]}f_nw_{l-2,1}+\frac{q^{-1}[l-1]}{[2][l]}
(f_{n-1}f_n-q^{-2}f_nf_{n-1})|{\bf e}_n\rangle\ot|(l-3){\bf e}_n\rangle \\
&\hspace{8cm}
+\frac{q^{l-1}}{[l]}f_{n-1}|{\bf e_n}\rangle\ot|(l-1){\bf e}_n\rangle\quad(l\ge3),\\
w_{l,k}&=(\frac1{[2]}f_{n-1}^2f_n-q^{-1}f_{n-1}f_nf_{n-1}+q^{-2}f_nf_{n-1}^2)w_{l-2,k-2}
+\frac{q^{2l-2k+1}}{[l-k+1][l-k+2]}f_{n-1}^2w_{l,k-2}\\
&\hspace{12.6cm}(2\le k\le l),
\end{align*}
we obtain the result.
\end{proof}

\begin{lemma}
Let $W_l$ be the vector subspace of $V^{\ot2}$ spanned by 
$|\sum_{j=1}^nk_j{\bf e}_j\rangle
\ot|\sum_{j=1}^nk'_j{\bf e}_j\rangle$ such that $\sum_{j=1}^n(k_j+k'_j)=l$. Then we have
\[
W_l=\sum_{0\le k\le l}U_q(\geh_{I\setminus\{0,n\}})w_{l,k}.
\]
\end{lemma}

\begin{proof}
As a $U_q(\geh_{I\setminus\{0,n\}})(=U_q(A_{n-1}))$-module $W_l$ is isomorphic to
$\bigoplus_{k=0}^lL(k{\bf e}_n)\ot L((l-k){\bf e}_n)$, where $L(\lambda)$ stands for
the irreducible highest weight $U_q(A_{n-1})$-module with highest weight $\lambda$.
By representation theory of $U_q(A_{n-1})$, $L(k{\bf e}_n)\ot L((l-k){\bf e}_n)$ is 
generated by the highest weight vectors of weight of the form 
$j{\bf e}_{n-1}+(l-j){\bf e}_n$
for some $0\le j\le\min(k,l-k)$. Hence, it is enough to show that any vector in $W_l$
of weight of the form $(l-j){\bf e}_n+j{\bf e}_{n-1}$ is generated by $w_{l,j}$
over $U_q(\geh_{\{n-1\}})(=U_q(sl_2))$. But it is a well-known fact from representation 
theory of $U_q(sl_2)$, namely, $L(a{\bf e}_n)\ot L(b{\bf e}_n)$ is generated by 
$|a{\bf e}_{n-1}\rangle\ot|b{\bf e}_n\rangle$, where $|a{\bf e}_{n-1}\rangle$
(resp. $|b{\bf e}_n\rangle$) is the lowest (resp. highest) weight vector.
\end{proof}

\begin{proof}[Proof of Prop. \ref{pr:irred}]
Suppose $\geh=D_{n+1}^{(2)}$ and let $W$ be a nonzero submodule of $V^{\ot2}$. 
In the proof of Prop. 12 of \cite{KO3}, we have shown that $W$ contains $v_l$ for any 
$l\ge0$.  Similarly, for $\geh=A_{2n}^{(2)}$ (resp. $C_n^{(1)}$), using Lemma 8
(resp. 10) of \cite{KO3}, we can show a nonzero submodule $W$ of $V^{\ot2}$
(resp. $V^{\epsilon_1}\ot V^{\epsilon_2}$) contains
$v_l^\epsilon$ for any $l\ge0,\epsilon=\pm$ (resp. $v_l^{\epsilon_1}$ for any
$l$ such that $(-1)^l=\epsilon_1\epsilon_2$). The claim now follows from Prop.
\ref{pr:w}.
\end{proof}

\section*{Acknowledgments}
The authors thank S. Sergeev for collaboration
in their previous work.
This work is supported by 
Grants-in-Aid for Scientific Research No.~23340007 and No.~24540203
from JSPS.
A.K. thanks the organizers of the 30th International 
Colloquium on Group Theoretical Methods in Physics at Ghent University 
during 14-18 July 2014 for hospitality.

\end{document}